\documentclass[
final
]{dmtcs-episciences}

\usepackage{amsmath}
\usepackage{amssymb,amsthm,url}
\usepackage{stmaryrd}
\usepackage{todonotes}
\usepackage{comment}
\usepackage{thm-restate}
\usepackage[numbers, sort&compress]{natbib} 
\usepackage{doi} 
\usepackage{bibentry}
\nobibliography*

\usepackage{tikz}

\hypersetup{
  colorlinks=true,
  linkcolor=black,
  citecolor=black,
  filecolor=black,
  urlcolor=[rgb]{0,0.1,0.5},
  pdftitle={Space Lower Bounds for Deterministic Self-Stabilizing Algorithms},
  pdfauthor={}
}

\newcommand{\A}{\mathcal A}

\newtheorem{definition}{Definition}
\newtheorem{lemma}{Lemma}

\title{Optimal Space Lower Bound for Deterministic Self-Stabilizing Leader Election Algorithms\thanks{Support by ANR ESTATE (ANR-16-CE25-0009-03) and ANR GrR (ANR-18-CE40-0032)}}

\author{L{\'e}lia Blin\affiliationmark{1}
\and 
Laurent Feuilloley\affiliationmark{2}
\and
Gabriel Le Bouder\affiliationmark{3}
}%

\affiliation{Sorbonne Université, LIP6 and Université d’Evry-Val-d’Essonne\\
Univ Lyon, CNRS, INSA Lyon, UCBL, Centrale Lyon, Univ Lyon 2, LIRIS\\
Sorbonne Université and INRIA}


\publicationdata{vol. 25:1}{2023}{5}{10.46298/dmtcs.9335}{2022-04-12; 2022-04-12; 2022-10-18; 2023-02-12}{2023-02-13}

\begin{document}
\maketitle



\begin{abstract}
Given a boolean predicate~$\Pi$ on labeled networks (e.g., proper coloring,
leader election, etc.), a self-stabilizing algorithm for~$\Pi$ is a distributed
algorithm that can start from any initial configuration of the network (i.e.,
every node has an arbitrary value assigned to each of its variables), and
eventually converge to a configuration satisfying~$\Pi$. It is  known that
leader election does not have a deterministic self-stabilizing algorithm using a
constant-size register at each node, i.e., for some networks, some of their
nodes must have registers whose sizes grow with the size~$n$ of the networks. On
the other hand,  it is also known that leader election can be solved by a
deterministic self-stabilizing algorithm using registers of $O(\log \log n)$
bits per node in any $n$-node bounded-degree network. We show that this latter
space complexity is optimal. Specifically, we prove that every deterministic
self-stabilizing algorithm solving leader election must use $\Omega(\log \log
n)$-bit per node registers in some $n$-node networks. In addition, we show that
our lower bounds go beyond leader election, and apply to all problems that
cannot be solved by anonymous algorithms.

\end{abstract}


\section{Introduction}

\subsection{Context}

Self-stabilization is a paradigm suited to asynchronous distributed systems
prone to transient failures. The occurrence of such a failure (e.g., memory
corruption) may move the system to an arbitrary configuration, and an algorithm
is self-stabilizing if it guarantees that whenever the system is in a
configuration that is illegal with respect to some given boolean
predicate~$\Pi$, the system returns to a legal configuration in finite time (and
remains in legal configuration as long as no other failures occur). In this
paper, we study self-stabilization in networks. The network is modeled as a
graph $G=(V,E)$, and we consider predicates defined on labeled graphs. For
instance, in proper $k$-coloring, a configuration is legal if every node is
labeled by a color $\{1,\dots,k\}$ that is different from the colors of all its
neighbors. Given a boolean predicate~$\Pi$, a self-stabilizing algorithm
for~$\Pi$ is a distributed algorithm enabling every node, given any input state,
to construct a label such that the resulting labeled graph satisfies~$\Pi$.

During the execution of a self-stabilizing algorithm, the nodes exchange
information along the links of the network, and this information is stored
locally at every node. Specifically, processes in a distributed system have two
types of memory: the \emph{persistent} memory, and the \emph{mutable} memory.
The persistent memory is used to store the identity of the process (e.g., its IP
address), its port numbers, and the code of the algorithm executed on the
process. Importantly, this section of the memory is not write enabled during the
execution of the algorithm. As a consequence it is less likely to be
corruptible, and most work in self-stabilization assumes that this part of the
memory is not subject to failures. The mutable memory is used to store the
variables used by the algorithm, and is subject to failures, that is, to the
corruption of these variables. The space complexity of a self-stabilizing
algorithm is the total size of all the variables used by the algorithm,
including those used to encode the output label of the node. For instance, the
space complexity of the algorithm for $k$-coloring is at least $\Omega(\log k)$
bits per node, for encoding the colors in $\{1,\dots,k\}$. The question
addressed in this paper is: under which circumstances is it possible to reach a
space complexity as low as the size of the labels? And if not, what is the
smallest space complexity that can be achieved?

Preserving small space complexity is indeed very much desirable, for several
reasons. First, it is expected that self-stabilizing algorithms offer some form
of universality, in the sense that they are executable on several types of
networks. Networks of sensors as used in IoT, as well as networks of robots as
used in swarm robotics, have the property to involve nodes with limited memory
capacity, and distributed algorithms of large space complexity may not be
executable on these types of networks. Second, a small space complexity is the
guarantee to consume a small bandwidth when nodes exchange information, thus
reducing the overhead due to link congestion~\cite{AdamekNT12}. In fact, a
self-stabilizing algorithm is never terminating, in the sense that it keeps
running in the background in case a failure occurs, for helping the system to
return to a legal configuration. Therefore, nodes may be perpetually exchanging
information, even after stabilization, and even when no faults occur. Limiting
the amount exchanged information, and thus, specially, the size of the
variables, is therefore of the utmost importance for optimizing time, and even
energy. Last but not least, increasing robustness against variable corruption
can be achieved by data replication~\cite{HermanP00}. This is, however, doable
only if the variables are reasonably small. Said otherwise, for a given memory
capacity, the smaller the space complexity the larger the robustness thanks to
data replication.

\subsection{Contributions of the paper}

In this paper, we focus on one of the arguably most important problems in the
context of distributed computing, namely \emph{leader election}. The objective
is to maintain a unique leader in the network, and to enable the network to
return to a configuration with a unique leader in case there are either zero or
more than one leader. Interestingly, encoding legal states consumes one single
bit at each node. Indeed, in a leader election, every node has a label with
value~0 or~1, and these labels form a legal configuration if there is one and
only one node with label~1. As a consequence, up to an additive constant, the
space complexity is exactly the space used to encode the variables of the
algorithm (which is not the case for other problems where the output itself uses
some non-trivial space to be encoded).

We establish the lower bound of $\Omega(\log \log n)$ bits per node for the
space complexity of leader election. This improves the only lower bound known so
far (see~\cite{BeauquierGJ99}), which states that leader election has
nonconstant space complexity, i.e., complexity $\omega(1)$, where $f=\omega(g)$
if $g(n)/f(n)\to 0$ when $n\to \infty$. More importantly, our bound matches the
best known upper bound on the space complexity of leader election, which is
$O(\log\log n)$ bits per node in bounded degree networks~\cite{BlinT18-b}, and
specially invalidates the folklore conjecture stating that leader election
is solvable using only $O(\log^*n)$ bits of mutable memory per node.

We obtain our lower bound by establishing an interesting connection between
self-stabilizing algorithms with small space complexity, and self-stabilizing
algorithms performing in homonym networks, that is, in networks in which
identifiers can be shared by several nodes. Homonym networks are an
intermediate situation between networks with unique identifiers, and anonymous
networks, which are networks in which nodes have no identifiers.
(Recall that space complexity counts solely the size of the mutable memory, and
does not include the immutable persistent memory where the identifiers are
stored). More specifically, the technical ingredient used for establishing our
results are the following. It is known that many self-stabilization problems,
including vertex coloring, leader election, spanning tree construction, etc.,
require that the nodes are provided with unique identifiers, for breaking
symmetry. Indeed, no algorithm can solve these problems in homonym networks
(under a standard distributed scheduler).
We show that, for any self-stabilizing algorithm in a network with unique node
identifiers, if the space complexity of the algorithm is too small, then the
algorithm does not have more power than a self-stabilizing algorithm running in
a homonym network. More precisely, let $\A$ be an algorithm in a
network with unique node identifiers, and let us assume that $\A$ has space
complexity $o(\log \log n)$ bits per node. Such a small space complexity does
not prevent $\A$ from exchanging identifiers between nodes, but they must be
transferred as a series of smaller pieces of information that are pipelined
along a link, each of size $o(\log \log n)$ bits. On the other hand, a node
cannot store the identifier of even just one of its neighbors. We show that,
with space complexity $o(\log \log n)$ bits per node, there exist graphs and
assignments of unique identifiers to the nodes of these graphs such that, in these
graphs and for these identifier-assignments, $\A$ has the same behavior as an
algorithm executed in these graphs but where identifiers can be
shared by several nodes (i.e. in the $k$-homonym version of these graphs).
We then show that no algorithms can solve leader election in these graphs
with duplicated identifiers, from which it follows that $\A$ cannot
solve leader election in these graphs with unique identifiers as long as its
space complexity is $o(\log \log n)$ bits per node.

Under a slightly different assumption on the distribution of identifiers, we
establish another equivalence with anonymous networks, instead of homonym
networks. Namely, with space complexity $o(\log \log n)$ bits per node, there
exist graphs and assignments of identifiers to the nodes of these graphs such
that $\A$ has the same behavior in these identified graphs and in the anonymous
version of these graphs (i.e in the absence of identifiers).

\subsection{Related work}

Space complexity of self-stabilizing algorithms has been extensively studied for
\emph{silent} algorithms, that is, algorithms that guarantee that the content of
the variables of every node does not change once the algorithm has reached a
legal configuration. For silent algorithms, Dolev and al. \cite{DolevGS99},
proved that finding the centers of a graph, electing a leader, and constructing
a spanning tree requires registers of $\Omega(\log n)$ bits per node. Silent
algorithms have later been related to a concept known as \emph{proof-labeling
scheme} (PLS)~\cite{KormanKP10}. Any lower bound on the size of the proofs in a
PLS for a predicate~$\Pi$ on labeled graph implies a lower bound on the size of
the registers for silent self-stabilizing algorithms solving~$\Pi$. A typical
example is the $\Omega(\log^2n)$-bit lower bound on the size of any PLS for
minimum weight spanning trees (MST)~\cite{KormanK07} , which implies the same
bound for constructing an MST in a silent self-stabilizing
manner~\cite{BlinF15}. Thanks to the tight connection between silent
self-stabilizing algorithms and proof-labeling schemes, the space complexity of
a vast collection of problems is known, for silent algorithms. (See
\cite{FeuilloleyF16} for more information on proof-labeling schemes.)

On the other hand, to our knowledge, the only lower bound on the space
complexity for general self-stabilizing algorithms (without the requirement of
being silent) that corresponds to our setting has been established by Beauquier,
Gradinariu and Johnen~\cite{BeauquierGJ99} who proved that registers of constant
size are not sufficient for leader election algorithms. Interestingly, the same
paper also contains several other space complexity lower bounds for models
different from ours -- e.g., anonymous networks, or harsher forms of asynchrony.
Although there are very few lower bounds for the model we consider, there exist
  several impossibility results for specific topologies.
In the case of anonymous networks, where nodes do not have a unique identifier,
Dijkstra~\cite{Dijkstra82} proved that there does not exist any self-stabilizing
algorithm for leader election, due to the general impossibility to break
symmetry in such networks. A slightly more powerful model than anonymous
networks is homonym networks, where identifiers might be shared by several
nodes. In~\cite{Delporte-GalletFT14} the authors propose a necessary and
sufficient condition on the number of distinct labels in bidirectional homonym
rings to solve leader election. Later,~\cite{AltisenDDDL20} generalizes the
previous result by establishing impossibility results for leader election in some
unidirectional homonym rings. Note that the impossibility results
of~\cite{Delporte-GalletFT14} and~\cite{AltisenDDDL20} are established for
general distributed algorithms, without any self-stabilization requirements.

The literature dealing with upper bounds is far richer. In particular,
\cite{BlinT18-a} recently presented a self-stabilizing leader election algorithm
using registers of $O(\log\log n)$ bits per node in $n$-node rings. This
algorithm was later generalized to networks with maximum degree $\Delta$, using
registers of $O(\log\log n+\log \Delta)$ bits per node~\cite{BlinT18-b}. It is
worth noticing that spanning tree construction and $(\Delta+1)$-coloring have
the same space complexity $O(\log\log n+\log \Delta)$ bits per node
~\cite{BlinT18-b}. Prior to these works, the best upper bound was a space
complexity $O(\log n)$ bits per node~\cite{BlinT18-a}, and it has then been
conjectured that, by some iteration of the technique enabling to reduce the
space complexity from $O(\log n)$ bits per node to $O(\log\log n)$ bits per
node, one could go all the way down to a space complexity of $O(\log^*n)$ bits
per node. Arguments in favor of this conjecture were that such successive
exponential improvements have been observed several times in distributed
computing. A prominent example is the time complexity of minimum spanning tree
construction in the congested clique model~\cite{LotkerPPP05, HegemanPPSS15,
GhaffariP16, Jurdzinski018}. Complexities $O(\log^*\!n)$ are not unknown in the
self-stabilizing framework~\cite{AwerbuchO94}, and it seemed at first that the
technique in~\cite{DolevGS99} could indeed be iterated (in a similar fashion as
in~\cite{BoczkowskiKN17}). Our result shows that this is not the case, and
somewhat closes the question of the space complexity of leader election.

A preliminary version of this paper appeared at OPODIS
2021~\cite{BlinFB21}. The previous version established a similar lower bound for
leader election algorithms, via the equivalence between algorithms with small
memory and algorithms for anonymous networks. The current version establishes
the same bound but with weaker assumptions on the identifier range and the
scheduler, via an equivalence between algorithms with small memory and
algorithms for homonym networks.


\section{Model and definitions}\label{sec:model-definitions}
In this paper, we are considering the \emph{state model} for self-stabilization~\cite{Dijkstra74}.
The asynchronous network is modeled as a simple $n$-node graph $ G = (V, E) $, where the set of the nodes $ V $ represents the processes, and the set of edges
$ E $ represents pairs of processes that can communicate directly with each other. Such pairs of processes are called \emph{neighbors}. The set of the neighbors of node $v$ is denoted $N(v)$.
Each node has local variables and a local algorithm.
The variables of a node are stored in its mutable memory, also called its \emph{register}.
In the state model, each node $v$ has read/write access to its register.
Moreover, in one atomic step, every node reads its own register and the registers of its neighbors, executes its local algorithm and updates its own register if necessary.
Note that the values of the variables of one node $v\in V$ are called the \emph{state} of $v$, and denoted by $S(v)$.

The core of this paper is establishing a link between networks with unique identifiers and weaker types of networks.
Let us consider one graph $G = (V, E)$.
We say that $G$ is an \emph{identified network} if each node $v\in V$ has a distinct identity, denoted by $ID(v)$.
For all $k\in \mathbb N^{*}$, we say that $G$ is a \emph{$k$-homonym network} if each node $v\in V$ has an identity $ID(v)$, and if for any identity $id$, there exists at most $k$ nodes in $V$ whose identity is $id$.
Remark that identified networks can be seen as $1$-homonym networks.
Finally, we say that $G$ is an \emph{anonymous network} if nodes do not have any identity, or equivalently, all nodes have the same identity.

For each adjacent edge, each node has access to a locally unique port number. No assumption is made on the consistency between port numbers on each node.
The mutable memory is the memory used to store the variables, while the immutable memory is used to store the identifier, the port numbers, and the code of the protocol.
As a consequence, the identity and the port numbers are non corruptible constants, and only the mutable memory is considered when computing the memory complexity because it corresponds to the memory readable by the neighbors of the nodes, and thus correspond to the information transmitted during the computation.
More precisely, an algorithm may refer to the identity, or to the port numbers, of the node, without the need to store them in the variables.

The output of the algorithm for a problem is carried through local variables of each node. The output of the problem may use all the local variables, or only a subset of them.
Indeed, the algorithm must have local variables that match the output of the problem, we call these variables the specification variables.
But the algorithm may also need some extra local variables that may be necessary to compute the specification variables.
For example, if we consider a silent BFS spanning tree construction, the specification variables are the variables dedicated to pointing out the parent in the BFS.
However, to respect the silent property, the algorithm needs in each node a variable dedicated to the identity of the root of the spanning tree and a variable dedicated to the distance from the root.
As a consequence, we define the \emph{specification} of problem $P$ as a description of the correct assignments of specification variables, for this specific problem $P$.

A \emph{configuration} is an assignment of values to all variables in the system, let us denote by $\Gamma$ the set of all the configurations.
A \emph{legal configuration} is a configuration $\gamma$ in $\Gamma$ that respects the specification of the problem, we denote by $\Gamma^*$ the set of legal configurations.
A local algorithm is a set of rules the node can apply, each rule is of the form \emph{<label>:<guard>$\rightarrow$<command>}. A \emph{guard} is a Boolean predicate that uses the local variables of the node and of its neighbors, and a \emph{command} is an assignment of variables.
A node is said to be \emph{enabled} if one of its guards is true and \emph{disabled} otherwise.

We consider an asynchronous network, the asynchrony of the system is modeled by an adversary called \emph{scheduler} or \emph{daemon}.
The scheduler chooses, at each step, which enabled nodes will execute a rule.
Several schedulers are introduced in the literature depending on their characteristics. Dubois and al. in \cite{DuboisT11} presented a complete overview of these schedulers.
Since we are interested in showing a lower bound, we aim for the least challenging scheduler.
Our lower bound is established under the \emph{central strongly fair scheduler}.
The central strongly fair scheduler activates exactly one enabled node at each step, which facilitates symmetry breaking.
A scheduler is strongly fair if, for every execution, a node which is infinitely often enabled is infinitely often activated.
The strong fairness property is captured by both weak fairness and unfairness, and the central activation is captured by locally central activation, and by distributed activation, which makes the central strongly fair scheduler the weakest scheduler in the literature.

A configuration $\gamma\in \Gamma$ is a legal configuration for the leader election problem if one single node is elected. More formally:

\begin{definition}[Leader election] \label{def:leader-election}
Leader election in $G=(V,E)$ is specified by a boolean variable $\ell_v$ at each node $v\in V$. A configuration $\{(v,\ell_v):v\in V\}$ is legal if there is a node $v\in V$ such that $\ell_v=true$, and for every other node $u\in V\smallsetminus\{v\}$, $\ell_u=false$.
\end{definition}

\section{Formal Statement of the Results}

\subsection{Lower bounds}

Our first result is an $\Omega(\log \log n)$ lower bound for leader election on the ring.

\begin{restatable}{theorem}{thmLEcycle}
  \label{thm:LE-cycle}
  Let $c\in \mathbb R, c > 1$. Every deterministic self-stabilizing algorithm solving leader election in the state model under a central strongly fair scheduler requires registers of size at least $\Omega(\log \log n)$ bits per node in $n$-node composite rings with unique identifiers in $[1,n\times c]$.
\end{restatable}

This bound improves the only lower bound known so far~\cite{BeauquierGJ99}, from $\omega(1)$ to $\Omega(\log \log n)$, and it is tight, as it matches the upper bound of \cite{BlinT18-a}, obtained in the same model and under a more challenging scheduler, the weakly fair distributed scheduler.
In particular, it invalidates the folklore conjecture stating that the aforementioned problems are solvable using only $O(\log^*\!n)$ memory.

\paragraph*{Optimality of the assumptions}

Our lower bound is actually optimal not only in terms of size, but also in terms of the assumptions we make on the setting.
More precisely, our theorem has three restrictions: it works for deterministic algorithms only, on rings with a non-prime number of nodes, and with identifiers in a large enough range.
Remark that we do not make any assumption on the scheduler.
Indeed, the central strongly fair scheduler is the least challenging scheduler, and represents therefore the most challenging hypothesis for an impossibility result.
We will now discuss why these limitations are actually necessary.

Randomization is a common tool for symmetry breaking, and our problem is one example.
Namely, \cite{ItkisL94} proved that using randomization, one can solve leader election using constant memory, which implies that our result cannot be generalized in that direction.

Theorem~\ref{thm:LE-cycle} only applies to algorithms that solve leader election on composite rings.
It cannot be applied to prime rings (that is, rings whose size $n$ is a prime number).
There is no hope to extend our result in that direction, since~\cite{ItkisLS95} builds a constant memory algorithm for leader election on anonymous prime rings, under a central scheduler.

Finally, and this is probably more surprising, we need to consider identities between 1 and $n\times c$, for $c>1$.
This is not an artifact of our proof: it is actually necessary for the result to hold.
Indeed, if the identifier range is $[1,n]$, then an algorithm may use the node with identifier 1 as a designated node, and have a special code for it.
Algorithms using such designated nodes are called \emph{semi-uniform algorithms} and they can achieve space complexity below our lower bound~\cite{DattaJPV00,Johnen97}.
Even without the possibility of having a designated node, one can take advantage of smaller identifier range: there actually exists a leader election algorithm for the ring which requires constant memory if the identities are in $[1,n+c]$ for a constant $c$~\cite{BeauquierGJ99}.

\paragraph*{A general result on the power of the identifiers}

Actually, our technique goes beyond the setting of Theorem~\ref{thm:LE-cycle}.
First, we do not need the harshest aspects of the self-stabilizing model which is that the initial configuration can be arbitrary.
If we start from an empty configuration, our technique still holds.
Second, the technique works for basically any problem that requires minimal symmetry breaking, not just leader election.
Third, as hinted above the type of scheduler is not really important. 
We prove the more general following theorem.

\begin{restatable}{theorem}{thmGeneralHomo}
  \label{thm:general-thm-homo}
  Let $c\in\mathbb R, c>1$ and let $k \in \mathbb N, k\geq 2$.
  Let us consider $(G_i)_{i\in\mathbb N}$ a class of graphs which contains graphs of arbitrary large size, multiple of $k$, such that $\Delta(G_i) \in o(\log n)$.

  If there exists a deterministic self-stabilizing algorithm $\mathcal A$ that solves a problem $\mathcal P$ under scheduler $\mathcal D$ in the state model and uses registers of size $o(\frac {\log\log n}{\Delta})$ on all identified networks $G_i$ with unique identifiers in $[1, n\times c]$, then
 $\mathcal A$ solves $\mathcal P$ under scheduler $\mathcal D$ in the state model on all sufficiently large $k$-homonym networks $G_i$.
\end{restatable}

Furthermore, if we consider a slightly different hypothesis on the range in which are chosen the identifiers, we can prove a similar result using anonymous networks.

\begin{restatable}{theorem}{thmGeneralAno}
  \label{thm:general-thm-ano}
  Let $c \in \mathbb R, c>1$.
  Let us consider $(G_i)_{i\in\mathbb N}$ a class of graphs which contains graphs of arbitrary large size, such that $\Delta(G_i) \in o(\log n)$.

  If there exists a deterministic self-stabilizing algorithm $\mathcal A$ that solves a problem $\mathcal P$ under scheduler $\mathcal D$ in the state model and uses registers of size $o(\frac {\log\log n}{\Delta})$ on all identified networks $G_i$ with unique identifiers in $[1, n^c]$, then
 $\mathcal A$ solves $\mathcal P$ under scheduler $\mathcal D$ in the state model on all sufficiently large anonymous networks $G_i$.
\end{restatable}

What our paper is really about is the power of identifiers in a scenario where very little space or communication is used.
Our core result is that if an algorithm uses less than $\Theta(\log \log n)$ bits per node, unique identifiers are useless in the sense that, in the worst case, the algorithm is not more powerful than it is on anonymous networks.
Remember that the Naor-Stockmeyer order-invariance theorem~\cite{NaorS95} states that in the LOCAL model, for local problems, constant-time algorithms that use the exact values of the identifiers are not more powerful than the order-invariant algorithm that only uses the relative ordering of the identifiers. In some sense our paper and \cite{NaorS95} have the same take-home message, in two different contexts: if you do not have enough resources, you cannot use the (full) power of the identifiers.

Theorem~\ref{thm:general-thm-ano} establishes that proving a $\Omega(\log \log n)$ lower bound for algorithms in identified networks boils down to proving that algorithms cannot solve the problem in anonymous networks.
This is useful, because indistinguishability arguments are easier to establish in anonymous networks than in identified networks.

Remark that, contrary to Theorem~\ref{thm:LE-cycle}, the statements of Theorems~\ref{thm:general-thm-homo} and~\ref{thm:general-thm-ano} do not refer to the scheduler.
In Theorem~\ref{thm:LE-cycle}, we prove one impossibility result and thus we must detail the settings for which it is obtained.
Theorems~\ref{thm:general-thm-homo} and~\ref{thm:general-thm-ano} are not directly impossibility results, but equivalence theorems, which are true regardless of the considered scheduler.

Finally, one aspect that is not explicit in the statement of Theorems~\ref{thm:general-thm-homo} and~\ref{thm:general-thm-ano} but follows from the proof, is that actually both results also hold if the nodes have inputs or outputs. In particular, our results apply to the semi-uniform setting where exactly one node has a special input.

\paragraph*{About the port number model}

Our general theorems, Theorems~\ref{thm:general-thm-homo} and~\ref{thm:general-thm-ano}, hold in the model where a node knows its port numbers but not the ones of its neighbors.
Intuitively, this implies that a node $u$ cannot specify that some piece of information is intended to the node of  port number $p$, because that node does not know it has been assigned port number $p$.

In some graphs, the port number assignment can be chosen in such a way that knowing the port number assignment of the neighbors does not help.
For example, in the cycle of Theorem~\ref{thm:LE-cycle}, we can arrange the port numbers such that every edge is assigned port number 1 by one endpoint, and port number 2 by the other.
This allows to generalize our first result to the model where a node knows both port numbers on every adjacent edge.

\paragraph*{Extension to larger range of identifiers}
As we said above, our theorems are optimal in all the assumptions we made.
The range in which are taken the identifiers is one of those assumptions: the smaller the range, the simpler it is to use characteristics of identifiers to break symmetries.
Reciprocally, if we suppose that the identifiers are taken in a larger range than what we assumed, then Theorems~\ref{thm:LE-cycle},~\ref{thm:general-thm-homo} and~\ref{thm:general-thm-ano} remain valid.

We can actually be a bit more specific.
The core of our proofs is about establishing a link between small memory, and the fact that several nodes have to behave exactly the same.
If we suppose that there are much more identifiers than what is stated in our theorems, then it is more likely to find several identifiers that will correspond to identical behaviors, and then we can hope finding a higher lower bound on memory.

This is actually true, and the following establishes how to proceed:
Suppose that the network is granted with unique identifiers in $[1, M(n)]$ with $n\times c = O(M(n))$ (\textit{resp.} $n^c = O(M(n))$).
Then we can substitute $\log\log n$ by $\log \log M(n)$ in Theorems~\ref{thm:LE-cycle} and~\ref{thm:general-thm-homo} (\textit{resp.} Theorem~\ref{thm:general-thm-ano}) and the new theorem is valid.

The proof given in Sections~\ref{sec:LE-cycle},~\ref{sec:general-homo} and~\ref{sec:general-ano} can easily be adapted to embrace the formulation above.
Yet, to avoid overloading the article with technical details, we only prove the initial version of Theorems~\ref{thm:LE-cycle},~\ref{thm:general-thm-homo} and~\ref{thm:general-thm-ano}.

\paragraph*{Validity under stronger forms of homonymy}
Theorem~\ref{thm:general-thm-homo} establishes an equivalence between algorithms with small memory, and algorithms for $k$-homonym networks, for any fixed value of $k\in \mathbb N$.
On the other hand, the technique used in the proof of Theorem~\ref{thm:LE-cycle} guarantees that for any $k\leq \sqrt n$, we can build a $k$-homonym network on which the algorithm behaves exactly the same way as in an identified network.
This is stronger than what is stated in Theorem~\ref{thm:general-thm-homo} since $k$ might depend on $n$.
Since the proofs of Theorem~\ref{thm:LE-cycle} and~\ref{thm:general-thm-homo} are basically the same, we can adapt the given proof of Theorem~\ref{thm:general-thm-homo} to make the result a bit more general.
Namely, we can replace "let $k\in \mathbb N, k\geq 2$ by "let $x\in \mathbb R, x<1$ and let $k(n) = n^x$".
This makes Theorem~\ref{thm:general-thm-homo} more general.

\paragraph*{Link with Dijkstra's impossibility result}

In a celebrated paper~\cite{Dijkstra82}, Dijkstra established, among other things, that one cannot break symmetry within anonymous composite rings.
This results holds under a central strongly fair scheduler.
Theorem~\ref{thm:LE-cycle} follows the same idea as Dijkstra's.
The core of the proof, and the statement of Theorems~\ref{thm:general-thm-homo} and~\ref{thm:general-thm-ano} is about proving that an algorithm with too little memory cannot fully use the power of identifiers.
The second part of the proof of Theorem~\ref{thm:LE-cycle} is basically a generalization of~\cite{Dijkstra82}: we prove that one cannot break symmetry in homonym composite rings.

Recall that assuming that the ring is composite is essential, due to the constant memory algorithm on anonymous prime rings of~\cite{ItkisLS95}.





\subsection{Intuition of the proofs}

\paragraph*{Challenge of lower bounds for non-silent algorithms}

Almost all lower bounds for self-stabilization are for silent algorithms, which are required to stay in the same configuration once they have stabilized.
These lower bounds are then about a static data structure, the stabilized solution.
The question boils down to establishing how much memory is needed to locally certify the global correctness of the solution, and this is well studied \cite{Feuilloley21}.

When we do not require that the algorithm should converge to one correct configuration, and stay there, there is no static structure on which we can reason.
It is then unclear how we can establish lower bounds.
One way is to think about invariants.
Consider a property that we can assume to hold in the initial configuration, and that is preserved by the computation (if it follows some memory requirement hypothesis).
If no correct output configuration has this property, then we can never reach a correct output configuration.

In our proof, the property that will be preserved is that every configuration is symmetric.
This can clearly be assumed for the original configuration, and we show that basically if the memory is limited then this is preserved at each step.
As the specification we use for leader election is that the leader should output 1, and the other nodes should output 0, then it is not possible that a proper output configuration is symmetric.

\paragraph*{Intuition on a toy problem}
Let us now give some intuition about why algorithms for identified networks with small memory are effective in $k$-homonym or anonymous networks as well.
The code of an algorithm $\mathcal{A}$ for identified networks may refer to the identifier of the node that is running it.
For example, a rule of the algorithm could be:
\begin{itemize}
  \item if the states of the current node and of its left and right neighbors are respectively $x$, $y$, and $z$, then: if the identifier is odd the new state is $a$, otherwise it is $b$.
\end{itemize}
Now suppose you have fixed an identifier, and you look at the rules for this fixed identifier.
In our example, if the identifier is 7, the rule becomes:
\begin{itemize}
  \item if the states of the current node and of its left and right neighbors are respectively $x$, $y$, and $z$, then: the new state is $a$.
\end{itemize}
This transformation can be done for any rule, thus, for an identifier $i$, we can get an algorithm~$\mathcal{A}_i$ specific to this identifier.
When we run $\mathcal{A}$ on every node, we can consider that every node, with some identifier $i$ is running $\mathcal{A}_i$.
Note that $\mathcal{A}_i$ does not refer in its code to the identifier.

The key observation is the following.
If the amount of memory an algorithm can use is very limited, then there is very limited number of different behaviors a node can have, especially if the code does not refer to the identifier.
Let us illustrate this point by studying an extreme example: a ring on which states have only one bit.
In this case the number of input configurations for a node is the set of views $(x,y,z)$ as above, with $x,y,z\in\{0,1\}$.
That is there are $2^3=8$ different inputs, thus the algorithm can be described with 8 different rules.
Since the output of the function is the new state, the output is also a single bit.
Therefore, there are at most $2^8=256$ different sets of rules,
that is 256 different possible behaviors for a node.
In other words, in this extreme case, each specific algorithm
$\mathcal{A}_i$ is equal to one of the behaviors of this list of 256 elements.
This implies that, if we take a ring with 257 nodes,
there exists two nodes with two distinct identifiers $i$ and $j$,
such that the specific algorithms $\mathcal{A}_i$ and $\mathcal{A}_j$ are equal.

This toy example is not strong enough for our purpose, as we want to argue about instances where the whole network is symmetric, not just two nodes, and where each node has a non-constant memory.
But the idea above can be strengthened to get our theorem.
The key is to use the hypothesis that the identifiers are taken from a large enough range.
As we have a pretty large palette of identifiers, we can always find, not only two,
but $n$ distinct identifiers in $[1,n\times c]$ that can be grouped such that the specific algorithm of all the nodes of the same group correspond to the exact same behavior.
In this case, it is as if the network was $k$-homonym, where $k$ is the size of each group.
If the identifiers are taken in $[1,n^c]$, we can even find $n$ distinct identifiers such that all the specific algorithms $\mathcal{A}_i$ correspond to the exact same behavior.
In this case it is as if the network was anonymous.

As soon as the network is homonym-like, we can start an execution from a symmetric configuration.
If the scheduler always activates all homonym nodes consecutively, then the execution contains an infinite number of symmetric configurations, and thus never stabilizes to a proper leader election execution.
This is presented in Figure~\ref{fig:LE-cycle}.

\begin{figure}[h]
  \begin{center}
        \includegraphics[scale=1]{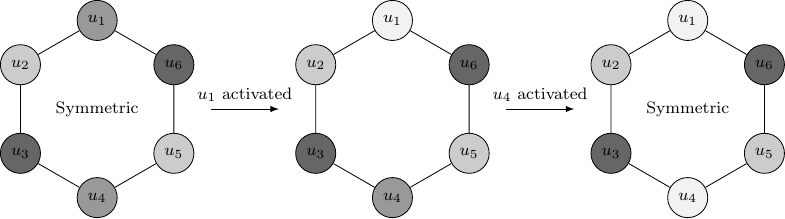}
    \caption{
      Example of two computing steps, with $n=6$, $k=2$, and
      $\A_{ID(u_1)} = \A_{ID(u_4)}$,
      $\A_{ID(u_2)} = \A_{ID(u_5)}$,
      $\A_{ID(u_3)} = \A_{ID(u_6)}$.
      Identical states are represented by identical shades of gray.
      If the states of the nodes are symmetric in the initial configuration, and if the behavior functions are symmetrically placed around the ring, then an execution can contain an infinite number of symmetric configurations.
      \label{fig:LE-cycle}
    }
  \end{center}
\end{figure}

Note that the larger the memory is, the more different behaviors there are, and the smaller
the set of identical specific algorithms we can find.
This trade-off implies that the construction works
as long as the memory is in $o(\log \log n)$.


\section{Proof of Theorem~\ref{thm:LE-cycle}}
\label{sec:LE-cycle}

In this section we prove our first result that establishes a lower bound on the
space complexity of self-stabilizing algorithms that solve leader election in
composite rings. \thmLEcycle*

Consider a ring of size $n$, and an algorithm for identified
networks $\A$ using $f(n)$ bits of memory per node to solve leader election
in composite rings. An algorithm can be seen as the function that
describes the behavior of the algorithm. This function takes an identifier, a
state for the node, a state for its left neighbor and a state for its right
neighbor, and gives the new state of the node. Formally:
\[
  \begin{array}{cccccccccr}
    \mathcal A: &[n\times c] &\times &\{0,1\}^{f(n)} &\times &\{0,1\}^{f(n)} &\times &\{0,1\}^{f(n)} &\to &\{0,1\}^{f(n)}\\
    &(ID &, &\text{state} &, &\text{left-state} &, &\text{right-state}) &\mapsto &\text{new-state}
  \end{array}
\]

Note that in general, we consider non-directed rings thus the nodes do not have
a global consistent definition for right and left. As we are dealing with a
lower bound with a worst-case on the port numbering, assuming such a consistent
orientation only makes the result stronger. Now we can consider that for every
identifier $i$, we have an algorithm of the form:
\[
  \begin{array}{cccccccr}
    \mathcal{A}_i: &\{0,1\}^{f(n)} &\times &\{0,1\}^{f(n)} &\times &\{0,1\}^{f(n)} &\to &\{0,1\}^{f(n)}\\
    &(\text{state} &, &\text{left-state} &, &\text{right-state}) &\mapsto &\text{new-state}
  \end{array}
\]
Thus a specific algorithm $\mathcal{A}_i$ boils down to a function of the form:  $\{0,1\}^{3f(n)} \to\{0,1\}^{f(n)}$.
Let us call such a function a \emph{behavior}, and let $\mathcal B_n$ be the sets of all behaviors.

Lemma~\ref{lem:card} counts the maximum number of distinct behaviors that can exist.
\begin{lemma}
  \label{lem:card}
  $| \mathcal B_n | = 2^{ f(n) \times 2^{3f(n)} }$
\end{lemma}
\begin{proof}
  The inputs are basically binary strings of length  $3f(n)$, thus there are $2^{3f(n)}$ possibilities for them.
  Similarly the number of possible outputs is $2^{f(n)}$.
  Thus the number of functions in $|\mathcal{B}_n|$ is
  $\left(2^{f(n)}\right)^{2^{3f(n)}}
  =2^{ f(n) \times 2^{3f(n)} }$.
\end{proof}

Lemma~\ref{lem:card} implies that the smaller $f$, the fewer different behaviors. Let us make this more concrete with Lemma~\ref{lem:domination}.
\begin{lemma}
  \label{lem:domination}
  If $f(n) \in o(\log\log n)$, then for every $n$ large enough, for every $k\leq \sqrt n$, we have $\frac{n(c-1)}{k-1} > |\mathcal B_n|$.
\end{lemma}
\begin{proof}
  Remark first that $\frac{n(c-1)}{k-1} \geq \sqrt n(c-1)$.
  Consider now the expression of $\sqrt n(c-1)$ and $|\mathcal B_n|$ after applying the logarithm twice:
  \[
    \begin{array}{rl}
      \log \log (\sqrt n(c-1)) &=\log(\frac12 \log n + \log(c-1)) \\
      &\sim \log\log n
    \end{array}
  \]
  \[
    \begin{array}{rl}
      \log \log (|\mathcal{B}_n|) &= \log \log \left(2^{ f(n) \times 2^{3f(n)} }\right)
      = \log \left( f(n) \times 2^{3f(n)} \right)
      =\log(f(n))+3f(n) \\
      &\sim 3f(n)
    \end{array}
  \]

  As the dominating term in the second expression is of order $f(n)\in o(\log\log n)$,
  asymptotically the first expression is larger.
  As $\log\log(\cdot)$ is an increasing positive function for large values,
  this implies that asymptotically $\frac{n(c-1)}{k-1} > |\mathcal B_n|$.
\end{proof}

Recall that our goal is to find $n$ different identifiers that can be grouped by
sets of size $k$ ($k$ being a divisor of $n$), such that in every group, all the
identifiers have the same corresponding behavior. If we consider $\varphi$ the
function that associates each identifier to the corresponding behavior, then it
boils down to finding $n$ distinct identifiers that can be grouped by $k$, such
that in each group, all the identifiers have the same image by $\varphi$. For
example, if $k=2$, $ID = \{1, 2, 3, 4, 5, 6, 7\}$, and $\mathcal B = \{b_1, b_2,
b_3\}$, we can have $\varphi(1) = \varphi(2) = \varphi(3) = b_1$, $\varphi(4) =
b_2$, and $\varphi(5) = \varphi(6) = \varphi(7) = b_3$. In that case, we can
form two sets of size $k=2$ such that all the elements of one set have the same
image by $\varphi$. We can for example consider $\{1,3\}$ and $\{6,7\}$.

Let us give a definition that formalizes that.
\begin{definition}
  Let $\varphi: A \to B$ be a function, and let $k\in \mathbb N$.
  We define the $k$-group number of $\varphi$, and denote by $t_k(\varphi)$ the maximum number of disjoint sets of size $k$ of elements of $A$, $S_1, S_2, \dots, S_{t_k(\varphi)}$ such that all the elements of the same set $S_i$ have the same image by $\varphi$.
\end{definition}

The precise value of $t_k(\varphi)$ depends on the specificities of $\varphi$.
Nevertheless, we can have a pretty good estimation of $t_k(\varphi)$ by comparing the respective sizes of $A$ and $B$.
Indeed, the larger $A$ is, the more chances we have to find sets of elements of $A$ that satisfy some property.
On the contrary, the larger $B$ is, the more possible images there are, and the harder it is to find elements of $A$ that are mapped to the same element of $B$.
Lemma~\ref{lem:two-collisions} establishes a generic lower bound on $t_k(\varphi)$.
\begin{lemma}
  \label{lem:two-collisions}
  Let $\varphi: A \to B$ be a function, and let $k\in \mathbb N$.
  We have $t_k(\varphi) \geq \frac{|A| - (k-1)|B|}{k}$.
\end{lemma}
\begin{proof}
  Let us first consider $b\in B$, and denote by $t_k^b(\varphi)$ the number of disjoint sets of size $k$ of elements of $A$ with image $b$ that can be formed.
  Intuitively, we have $t_k^b(\varphi) = \lfloor \frac{|\varphi^{-1}(b)|}k \rfloor$.
  By definition, we also have
  \[
    t_k(\varphi) = \sum_{b\in B} t_k^b(\varphi) = \sum_{b\in B} \lfloor \frac{|\varphi^{-1}(b)|}k \rfloor.
\]
  Recall the inequality, true for any integer $m$: $\lfloor \frac mk \rfloor \geq \frac{m-(k-1)}k$.
  Thus, we deduce:
  \[
    t_k(\varphi) \geq \sum_{b\in B} \frac1k(|\varphi^{-1}(b)| -(k-1)) \geq \frac1k \sum_{b\in B} (|\varphi^{-1}(b)| -(k-1)) \geq \frac1k \bigl(\sum_{b\in B} (|\varphi^{-1}(b)|) - (k-1)|B|\bigr) \\
  \]
  Since $\varphi$ is a function from $A$ to $B$ we have $\sum_{b\in B} (|\varphi^{-1}(b)|) = |A|$ and thus we conclude $t_k(\varphi) \geq \frac{|A| - (k-1)|B|}k$.
\end{proof}

The next lemma shows that if $\frac{n(c-1)}{k-1} > |\mathcal B_n|$, then for every $k\leq \sqrt n$
we can find $n/k$ disjoint sets of $k$ identifiers taken in $[1,n\times c]$ such that all the identifiers of each set have the same corresponding behavior.

\begin{lemma}\label{lem:exists_indices}
  Let $\varphi_n:[1,n\times c] \to \mathcal B_n$ be the function that associates each identifier $i$ to its corresponding behavior $\mathcal A_i$, and let $k\leq \sqrt n$.
  If $\frac{n(c-1)}{k-1} > |\mathcal B_n|$, then $t_k(\varphi_n) \geq n/k$.
\end{lemma}
\begin{proof}
  According to Lemma~\ref{lem:two-collisions} we have $t_k(\varphi_n) \geq \frac1k (nc - (k-1)|\mathcal B_n|)$, and by hypothesis this means that $t_k(\varphi_n) \geq \frac1k (nc - n(c-1))$, and thus $t_k(\varphi_n) \geq n/k$.
\end{proof}

Combining the three lemmas we get that, if $f(n) \in o(\log\log n)$,
then for any large enough~$n$, for any $k\leq \sqrt n$, we can find $n$ different identifiers in $[1,n\times c]$
which can be grouped in sets of size $k$ such that all identifiers of the same set have the exact same behavior.

Now, consider a large enough composite ring of size $n$, and let $k$ be a divisor of $n$ such that $1<k\leq \sqrt n$.
Let us consider $n$ identifiers that can be grouped into sets $S_0, S_1, \dots, S_{n/k - 1}$ of size $k$, as explained above.
Let us place these identifiers on the ring such that the $k$ identifiers of the same set are placed every $n/k$-th node.
As we are in a self-stabilizing scenario, we can choose from which configuration we start.
We actually do not need an intricate configuration: we start from a symmetric configuration, where all the nodes with identifier in the same set have the same state.
This cannot be a proper leader election output, because in leader election (with our specification) a leader and a non-leader have different output. Therefore, at least one node $v$ is enabled.
Let us denote by $S_v$ the set which contains $v$'s identifier.
Note that all the nodes whose identifier is in $S_v$ see the same states for themselves, and for both their neighbors.
Thus, since all of them also have the same behavior, if one is enabled, then the others are too, and if activated they will execute the same rule.
Now, the central scheduler activates once each one of the $k$ nodes whose identifier is in $S_v$, which takes $k$ steps.
The action of one node only changes its own state, and since these $k$ activated nodes are placed at a distance at least $2$ from each other, none of them will see the changes made by the others, so they all perform the same action when activated.
Thus, after these $k$ computing steps, the ring is once again in a symmetric configuration.

We can iterate this argument forever, as long as the scheduler consecutively activates nodes from the same set, which a central scheduler can do.
Furthermore, these activations can be made strongly fair, by choosing as new activated set $S_i$ the set of enabled nodes that have not been activated for the longest time, for example.
In other words, there exists an execution of $\A$ under the central strongly fair scheduler such that the system is infinitely often in symmetric configurations.
Therefore, the network never stabilizes in an execution that satisfies leader election.
This proves Theorem~\ref{thm:LE-cycle}.


%


\section{Proof of Theorem~\ref{thm:general-thm-homo}}
\label{sec:general-homo}

In this section, we prove our general result that relates the power of algorithms with small memory on identified networks to the power of algorithms on $k$-homonym networks.
\thmGeneralHomo*
The proof of Theorem~\ref{thm:general-thm-homo} follows the same idea as the one of Theorem~\ref{thm:LE-cycle}, but in a more general way.
We generalize in two directions: first we formally prove an equivalence with $k$-homonym networks, and second we consider a much larger class of graphs than composite rings.
Consider a graph $G_i$ of degree $\Delta$, and an algorithm $\mathcal{A}$ using $f(n) \in o(\frac{\log\log n}{\Delta})$ bits of memory per node to solve an arbitrary problem on identified networks.

As before an algorithm can be seen as a function that maps the view of the node to an output. But now as the degree is not the same for every node, we have one function for every degree.
This not very convenient for us, so we take another point of view.
We consider a function that takes as input: the node's identifier, its state, $\Delta$ other states, and a number $\delta$.

\[
  \begin{array}{cccccccccr}
    \mathcal A: &[n\times c] &\times &\{0,1\}^{f(n)} &\times & \left(\{0,1\}^{f(n)}\right)^{\Delta} & \times & [\Delta] &\to &\{0,1\}^{f(n)}\\
       &(ID &, &\text{state} &, &\text{$\Delta$ states} &, &\text{degree}) &\mapsto &\text{new-state}
  \end{array}
\]

The idea is that, to get its output, node $v$ will first feed its own degree $\Delta_v$ as the last input of the function, and then feed its identifier $ID(v)$, its own state $S(v)$, and the states of its $\Delta_v$ first neighbors as the $\Delta_v+2$ first fields. The other fields are left blank. The output of the function is the output of the algorithm.

Note that this corresponds to the single-port setting: in the function, each neighboring state is identified with the local port number associated with that neighbor. In particular, it is not a set of neighboring states. But it is not a double-port setting: which port a neighbor assigns to a node is unknown.

Now if we fix the identifier $i$, we get:
\[
  \begin{array}{cccccccccr}
    \mathcal A_i: &\{0,1\}^{f(n)} &\times & \left(\{0,1\}^{f(n)}\right)^{\Delta} & \times & [\Delta] &\to &\{0,1\}^{f(n)}\\
       &(\text{state} &, &\text{$\Delta$ states} &, &\text{degree}) &\mapsto &\text{new-state}
  \end{array}
\]

Now the equivalent of Lemma~\ref{lem:card} is that the number of such behaviors~$\mathcal{A}_i$ is:
\[ 
|\mathcal{B}_n| \leq \left(
2^{f(n)}
\right)^{2^{(\Delta+1)f(n)}\Delta}
\]

We bound the double logarithm of this expression:
\begin{align*}
  \log \log \left[ \left( 2^{f(n)} \right) ^ {2^{(\Delta+1)f(n)}\Delta} \right]
  &= \log \left[ f(n) 2^{(\Delta+1)f(n)}\Delta \right] \\
  &= (\Delta+1)f(n) + \log f(n) + \log \Delta\\
  & \in o(\log \log n)
\end{align*}

Thus, since $k$ is now a constant, we can compare $\log\log( \frac{n(c-1)}{k-1} )$ to $\log \log |\mathcal{B}_n|$ to conclude that, for any large enough $n$, $\frac{n(c-1)}{k-1} > |\mathcal B_n|$.
If we consider $\varphi_n: [1,n\times c] \to \mathcal B_n$ we deduce similarly to Lemma~\ref{lem:exists_indices} that for any large enough $n$, $t_k(\varphi_n) \geq n/k$.

Let us now consider a large enough graph $G$ with size $n$ a multiple of $k$ and with degree at most $\delta(n) \in o(\log n)$.
Let us consider $S_1, S_2, \dots, S_{n/k}$, disjoint sets of $k$ identifiers, such that all the identifiers of the same set have the same image by $\varphi_n$.
For each set $S_i$, let us pick one specific identifier, $ID_i$.
Now consider $G_h$ a $k$-homonym network with topology $G$ and $k$ occurrences of each $ID_i$,
and $G_{ID}$ an identified network similar to $G_h$, where for each $i$, occurrences of identifier $ID_i$ are substituted by one of each value of $S_i$.
By construction, any execution of $\A$ on $G_h$ is also an execution of $\A$ on $G_{ID}$, and since by hypothesis, $\A$ solves problem $P$ under scheduler $\mathcal D$ on $G_{ID}$, it also solves $P$ under scheduler $\mathcal D$ on $G_h$.
This completes the proof of Theorem~\ref{thm:general-thm-homo}.

\section{Proof of Theorem~\ref{thm:general-thm-ano}}
\label{sec:general-ano}

In this section, we prove our general result that relates the power of algorithms with small memory on identified networks, to the power of algorithms on anonymous networks.
\thmGeneralAno*
The proof of Theorem~\ref{thm:general-thm-ano} is basically the same as the proof of Theorem~\ref{thm:general-thm-homo}.
The difference is that if we suppose that the range of the identifiers is larger, then it is easier to find different identifiers that are mapped to the same behavior.
If we suppose that the range is polynomial in $n$, then for any large enough $n$, we can find $n$ different identifiers that all correspond to the same behavior.
Consider a graph $G$ of degree $\Delta$, and an algorithm for identified networks $\mathcal{A}$ using $f(n) \in o(\frac{\log\log n}{\Delta})$ bits of memory per node to solve an arbitrary problem.

The prototype of the function that describes the algorithm is now:
\[
  \begin{array}{cccccccccr}
    \mathcal A: &[n^c] &\times &\{0,1\}^{f(n)} &\times & \left(\{0,1\}^{f(n)}\right)^{\Delta} & \times & [\Delta] &\to &\{0,1\}^{f(n)}\\
       &(ID &, &\text{state} &, &\text{$\Delta$ states} &, &\text{degree}) &\mapsto &\text{new-state}
  \end{array}
\]

If we fix the identifier $i$, we get:
\[
  \begin{array}{cccccccccr}
    \mathcal A_i: &\{0,1\}^{f(n)} &\times & \left(\{0,1\}^{f(n)}\right)^{\Delta} & \times & [\Delta] &\to &\{0,1\}^{f(n)}\\
       &(\text{state} &, &\text{$\Delta$ states} &, &\text{degree}) &\mapsto &\text{new-state}
  \end{array}
\]

We still have the following inequalities:
\[
|\mathcal{B}_n| \leq \left(
2^{f(n)}
\right)^{2^{(\Delta+1)f(n)}\Delta}
\]

and
\[
\log \log |\mathcal B_n| \in  o(\log \log n)
\]

By comparing $\log \log (n^{c-1})$ and $\log\log |\mathcal B_n|$, we can prove that, for any large enough $n$, $n^{c-1} > |\mathcal{B}_n|$.
Finally, let us consider the function $\varphi_n: [1,n^c] \to \mathcal B_n$.
We deduce from Lemma~\ref{lem:two-collisions} that for any large enough $n$, $t_n(\varphi_n) \geq \frac{n^c - (n-1)n^{c-1}}{n} \geq n^{c-2} > 0$, and since $t_n(\varphi_n)$ is an integer, we have $t_n(\varphi_n) \geq 1$.
In other words, there exist $n$ distinct identifiers that are all mapped to the same behavior.

Let us now consider a large enough graph $G$ with degree at most $\delta(n) \in o(\log n)$.
Let us consider $n$ identifiers $ID_1, \cdots ID_n$ that all have the same image by $\varphi_n$.
Now consider $G_a$ an anonymous network with topology $G$ and where each node has identifier $ID_1$,
and $G_{ID}$ an identified network similar to $G_a$, where the identifiers of the nodes are $ID_1, \cdots ID_n$.
By construction, any execution of $\A$ on $G_a$ is also an execution of $\A$ on $G_{ID}$, and since by hypothesis, $\A$ solves problem $P$ under scheduler $\mathcal D$ on $G_{ID}$, it also solves $P$ under scheduler $\mathcal D$ on $G_a$.
This completes the proof of Theorem~\ref{thm:general-thm-ano}.


\section{Conclusion}

In this paper, we have established a lower bound $\Omega(\log\log n)$ bits per node on the size of the registers for self-stabilizing algorithms solving leader election in the state model.
This bound matches the upper bound $O(\log\log n)$ bits per node for bounded-degree graphs~\cite{BlinT18-b}.

Yet, for arbitrary graphs,~\cite{BlinT18-b} requires an additional space in $O(\log \Delta)$ bits per node.
An interesting problem would be to find whether this additional term in $O(\log \Delta)$ is necessary for graphs where the degree is not bounded.


\DeclareUrlCommand{\Doi}{\urlstyle{same}}
\renewcommand{\doi}[1]{\href{https://doi.org/#1}{\footnotesize\sf doi:\Doi{#1}}}

\bibliographystyle{plainnat}

\bibliography{biblio-lower-bound.bib}{}

\begin{thebibliography}{28}
\providecommand{\natexlab}[1]{#1}
\providecommand{\url}[1]{\texttt{#1}}
\expandafter\ifx\csname urlstyle\endcsname\relax
  \providecommand{\doi}[1]{doi: #1}\else
  \providecommand{\doi}{doi: \begingroup \urlstyle{rm}\Url}\fi

\bibitem[Adamek et~al.(2012)Adamek, Nesterenko, and Tixeuil]{AdamekNT12}
Jordan Adamek, Mikhail Nesterenko, and S{\'{e}}bastien Tixeuil.
\newblock Evaluating practical tolerance properties of stabilizing programs
  through simulation: The case of propagation of information with feedback.
\newblock In \emph{Stabilization, Safety, and Security of Distributed Systems -
  14th International Symposium, {SSS} 2012}, pages 126--132, 2012.
\newblock \doi{10.1007/978-3-642-33536-5\_13}.

\bibitem[Altisen et~al.(2020)Altisen, Datta, Devismes, Durand, and
  Larmore]{AltisenDDDL20}
Karine Altisen, Ajoy~K. Datta, St{\'{e}}phane Devismes, Ana{\"{\i}}s Durand,
  and Lawrence~L. Larmore.
\newblock Election in unidirectional rings with homonyms.
\newblock \emph{J. Parallel Distributed Comput.}, 146:\penalty0 79--95, 2020.
\newblock \doi{10.1016/j.jpdc.2020.08.004}.
\newblock URL \url{https://doi.org/10.1016/j.jpdc.2020.08.004}.

\bibitem[Awerbuch and Ostrovsky(1994)]{AwerbuchO94}
Baruch Awerbuch and Rafail Ostrovsky.
\newblock Memory-efficient and self-stabilizing network {RESET} (extended
  abstract).
\newblock In \emph{13th Annual {ACM} Symposium on Principles of Distributed
  Computing, {PODC} 1994}, pages 254--263, 1994.
\newblock \doi{10.1145/197917.198104}.

\bibitem[Beauquier et~al.(1999)Beauquier, Gradinariu, and
  Johnen]{BeauquierGJ99}
Joffroy Beauquier, Maria Gradinariu, and Colette Johnen.
\newblock Memory space requirements for self-stabilizing leader election
  protocols.
\newblock In \emph{18th Annual {ACM} Symposium on Principles of Distributed
  Computing, {PODC} 1999}, pages 199--207, 1999.
\newblock \doi{10.1145/301308.301358}.

\bibitem[Blin and Fraigniaud(2015)]{BlinF15}
L{\'{e}}lia Blin and Pierre Fraigniaud.
\newblock Space-optimal time-efficient silent self-stabilizing constructions of
  constrained spanning trees.
\newblock In \emph{35th {IEEE} International Conference on Distributed
  Computing Systems, {ICDCS} 2015}, pages 589--598, 2015.
\newblock \doi{10.1109/ICDCS.2015.66}.

\bibitem[Blin and Tixeuil(2018{\natexlab{a}})]{BlinT18-a}
L{\'{e}}lia Blin and S{\'{e}}bastien Tixeuil.
\newblock Compact deterministic self-stabilizing leader election on a ring: the
  exponential advantage of being talkative.
\newblock \emph{Distributed Computing}, 31\penalty0 (2):\penalty0 139--166,
  2018{\natexlab{a}}.
\newblock \doi{10.1007/s00446-017-0294-2}.

\bibitem[Blin and Tixeuil(2018{\natexlab{b}})]{BlinT18-b}
L{\'{e}}lia Blin and S{\'{e}}bastien Tixeuil.
\newblock Compact self-stabilizing leader election for general networks.
\newblock In \emph{{LATIN} 2018: Theoretical Informatics - 13th Latin American
  Symposium}, pages 161--173, 2018{\natexlab{b}}.
\newblock \doi{10.1007/978-3-319-77404-6\_13}.

\bibitem[Blin et~al.(2021)Blin, Feuilloley, and Bouder]{BlinFB21}
L{\'{e}}lia Blin, Laurent Feuilloley, and Gabriel~Le Bouder.
\newblock Optimal space lower bound for deterministic self-stabilizing leader
  election algorithms.
\newblock In \emph{25th International Conference on Principles of Distributed
  Systems, {OPODIS} 2021}, volume 217 of \emph{LIPIcs}, pages 24:1--24:12,
  2021.
\newblock \doi{10.4230/LIPIcs.OPODIS.2021.24}.

\bibitem[Boczkowski et~al.(2017)Boczkowski, Korman, and Natale]{BoczkowskiKN17}
Lucas Boczkowski, Amos Korman, and Emanuele Natale.
\newblock Minimizing message size in stochastic communication patterns: Fast
  self-stabilizing protocols with 3 bits.
\newblock In \emph{28th Annual {ACM-SIAM} Symposium on Discrete Algorithms,
  {SODA} 2017}, pages 2540--2559, 2017.
\newblock \doi{10.1137/1.9781611974782.168}.

\bibitem[Datta et~al.(2000)Datta, Johnen, Petit, and Villain]{DattaJPV00}
Ajoy~Kumar Datta, Colette Johnen, Franck Petit, and Vincent Villain.
\newblock Self-stabilizing depth-first token circulation in arbitrary rooted
  networks.
\newblock \emph{Distributed Computing}, 13\penalty0 (4):\penalty0 207--218,
  2000.
\newblock \doi{10.1007/PL00008919}.

\bibitem[Delporte-Gallet et~al.(2014)Delporte-Gallet, Fauconnier, and
  Tran-The]{Delporte-GalletFT14}
Carole Delporte-Gallet, Hugues Fauconnier, and Hung Tran-The.
\newblock Leader election in rings with homonyms.
\newblock In \emph{NETYS}, pages 9--24, 2014.
\newblock URL \url{https://doi.org/10.1007/978-3-319-09581-3_2}.

\bibitem[Dijkstra(1974)]{Dijkstra74}
Edsger~W. Dijkstra.
\newblock Self-stabilizing systems in spite of distributed control.
\newblock \emph{Commun. {ACM}}, 17\penalty0 (11):\penalty0 643--644, 1974.
\newblock \doi{10.1145/361179.361202}.

\bibitem[Dijkstra(1982)]{Dijkstra82}
Edsger~W. Dijkstra.
\newblock \emph{Self-Stabilization in Spite of Distributed Control}, pages
  41--46.
\newblock Springer New York, New York, NY, 1982.
\newblock ISBN 978-1-4612-5695-3.
\newblock \doi{10.1007/978-1-4612-5695-3_7}.

\bibitem[Dolev et~al.(1999)Dolev, Gouda, and Schneider]{DolevGS99}
Shlomi Dolev, Mohamed~G. Gouda, and Marco Schneider.
\newblock Memory requirements for silent stabilization.
\newblock \emph{Acta Inf.}, 36\penalty0 (6):\penalty0 447--462, 1999.
\newblock \doi{10.1007/s002360050180}.

\bibitem[Dubois and Tixeuil(2011)]{DuboisT11}
Swan Dubois and S{\'{e}}bastien Tixeuil.
\newblock A taxonomy of daemons in self-stabilization.
\newblock {\footnotesize \sf{arXiv:
  \href{http://arxiv.org/abs/1110.0334}{1110.0334}}}, 2011.

\bibitem[Feuilloley(2021)]{Feuilloley21}
Laurent Feuilloley.
\newblock Introduction to local certification.
\newblock \emph{Discret. Math. Theor. Comput. Sci.}, 23\penalty0 (3), 2021.
\newblock \doi{10.46298/dmtcs.6280}.

\bibitem[Feuilloley and Fraigniaud(2016)]{FeuilloleyF16}
Laurent Feuilloley and Pierre Fraigniaud.
\newblock Survey of distributed decision.
\newblock \emph{Bulletin of the {EATCS}}, 119, 2016.
\newblock {\footnotesize\sf{url:
  \href{http://bulletin.eatcs.org/index.php/beatcs/article/view/411/391}{bulletin.eatcs.org
  link}}}{\footnotesize , \sf{arXiv:
  \href{https://arxiv.org/abs/1606.04434}{1606.04434}}}.

\bibitem[Ghaffari and Parter(2016)]{GhaffariP16}
Mohsen Ghaffari and Merav Parter.
\newblock {MST} in log-star rounds of congested clique.
\newblock In \emph{2016 {ACM} Symposium on Principles of Distributed Computing,
  {PODC} 2016}, pages 19--28, 2016.
\newblock \doi{10.1145/2933057.2933103}.

\bibitem[Hegeman et~al.(2015)Hegeman, Pandurangan, Pemmaraju, Sardeshmukh, and
  Scquizzato]{HegemanPPSS15}
James~W. Hegeman, Gopal Pandurangan, Sriram~V. Pemmaraju, Vivek~B. Sardeshmukh,
  and Michele Scquizzato.
\newblock Toward optimal bounds in the congested clique: Graph connectivity and
  {MST}.
\newblock In \emph{2015 {ACM} Symposium on Principles of Distributed Computing,
  {PODC} 2015}, pages 91--100, 2015.
\newblock \doi{10.1145/2767386.2767434}.

\bibitem[Herman and Pemmaraju(2000)]{HermanP00}
Ted Herman and Sriram~V. Pemmaraju.
\newblock Error-detecting codes and fault-containing self-stabilization.
\newblock \emph{Inf. Process. Lett.}, 73\penalty0 (1-2):\penalty0 41--46, 2000.
\newblock \doi{10.1016/S0020-0190(99)00164-7}.

\bibitem[Itkis and Levin(1994)]{ItkisL94}
Gene Itkis and Leonid~A. Levin.
\newblock Fast and lean self-stabilizing asynchronous protocols.
\newblock In \emph{35th Annual Symposium on Foundations of Computer Science
  {FOCS} 1994}, pages 226--239, 1994.
\newblock \doi{10.1109/SFCS.1994.365691}.

\bibitem[Itkis et~al.(1995)Itkis, Lin, and Simon]{ItkisLS95}
Gene Itkis, Chengdian Lin, and Janos Simon.
\newblock Deterministic, constant space, self-stabilizing leader election on
  uniform rings.
\newblock In \emph{Distributed Algorithms, 9th International Workshop, {WDAG}
  '95}, pages 288--302, 1995.
\newblock \doi{10.1007/BFb0022154}.

\bibitem[Johnen(1997)]{Johnen97}
Colette Johnen.
\newblock Memory efficient, self-stabilizing algorithm to construct {BFS}
  spanning trees.
\newblock In \emph{Proceedings of the Sixteenth Annual {ACM} Symposium on
  Principles of Distributed Computing, {PODC} 97}, page 288, 1997.
\newblock \doi{10.1145/259380.259508}.

\bibitem[Jurdzinski and Nowicki(2018)]{Jurdzinski018}
Tomasz Jurdzinski and Krzysztof Nowicki.
\newblock {MST} in \emph{O}(1) rounds of congested clique.
\newblock In \emph{29th Annual {ACM-SIAM} Symposium on Discrete Algorithms,
  {SODA} 2018}, pages 2620--2632, 2018.
\newblock \doi{10.1137/1.9781611975031.167}.

\bibitem[Korman and Kutten(2007)]{KormanK07}
Amos Korman and Shay Kutten.
\newblock Distributed verification of minimum spanning trees.
\newblock \emph{Distributed Computing}, 20\penalty0 (4):\penalty0 253--266,
  2007.
\newblock \doi{10.1007/s00446-007-0025-1}.

\bibitem[Korman et~al.(2010)Korman, Kutten, and Peleg]{KormanKP10}
Amos Korman, Shay Kutten, and David Peleg.
\newblock Proof labeling schemes.
\newblock \emph{Distributed Computing}, 22\penalty0 (4):\penalty0 215--233,
  2010.
\newblock \doi{10.1007/s00446-010-0095-3}.

\bibitem[Lotker et~al.(2005)Lotker, Patt{-}Shamir, Pavlov, and
  Peleg]{LotkerPPP05}
Zvi Lotker, Boaz Patt{-}Shamir, Elan Pavlov, and David Peleg.
\newblock Minimum-weight spanning tree construction in \emph{O}(log log
  \emph{n}) communication rounds.
\newblock \emph{{SIAM} J. Comput.}, 35\penalty0 (1):\penalty0 120--131, 2005.
\newblock \doi{10.1137/S0097539704441848}.

\bibitem[Naor and Stockmeyer(1995)]{NaorS95}
Moni Naor and Larry~J. Stockmeyer.
\newblock What can be computed locally?
\newblock \emph{{SIAM} J. Comput.}, 24\penalty0 (6):\penalty0 1259--1277, 1995.
\newblock \doi{10.1137/S0097539793254571}.

\end{thebibliography}
\end{document}